\documentclass[11pt,letterpaper]{article}
\usepackage[margin=1in]{geometry}
\usepackage{amsmath,amssymb,amsfonts,setspace,url,mathrsfs}
\usepackage{amsthm}
\usepackage{float}
\usepackage{bigstrut,array,multirow}
\usepackage{color}
\usepackage{hyperref}
\usepackage[capitalise]{cleveref}
\usepackage{tikz}
\usepackage{mathpazo}
\usepackage[boxed,linesnumbered,noend]{algorithm2e}
\DontPrintSemicolon
\let\oldnl\nl
\newcommand{\nonl}{\renewcommand{\nl}{\let\nl\oldnl}}
\usepackage{algpseudocode}
\usepackage{braket}

\newcommand{\Gg}{\mathcal{G}}
\newcommand{\abs}[1]{\vert #1 \vert}
\newcommand{\neighbor}[1]{\mathcal N(#1)}
\newcommand{\Comp}{\mathbb C}

\newcommand{\poly}{\mathrm{poly}}

\newcommand{\Aa}{\mathtt{DistQWalk}}
\newcommand{\Bb}{\mathtt{Traversal}}

\newcommand{\Rr}{\mathsf{R}}
\newcommand{\Hh}{\mathcal{H}}
\newcommand{\mlength}{b}
\newcommand{\entr}{\mathtt{source}}
\newcommand{\exit}{\mathtt{target}}
\newcommand{\mess}{x}
\newcommand{\qmess}{\varphi}

\newcommand{\Roots}{\mathrm{Roots}}
\newcommand{\id}{\mathsf{id}}

\newtheorem{theorem}{Theorem}[section]

\newtheorem{proposition}{Proposition}[section]
\newtheorem{definition}{Definition}[section]

\newtheorem{lemma}{Lemma}[section]
\newtheorem{game}{Game}

\title{Exponential Quantum Advantage for Message Complexity \\in Distributed Algorithms}
 \author{
 Fran{\c c}ois Le Gall\\
 Nagoya University
 \and Ma{\"e}l Luce\\
 Nagoya University
 \and Joseph Marchand\\
 Ecole normale sup\'erieure Paris-Saclay
 \and Mathieu Roget\\
 Universit\'e Aix-Marseille}
\date{}
\begin{document}
\maketitle
\begin{abstract}
    We investigate how much quantum distributed algorithms can outperform classical distributed algorithms with respect to the \emph{message complexity} (the overall amount of communication used by the algorithm). Recently, Dufoulon, Magniez and Pandurangan (PODC 2025) have shown a polynomial quantum advantage for several tasks such as leader election and agreement. In this paper, we show an exponential quantum advantage for a fundamental task: routing information between two specified nodes of a network. We prove that for the family of ``welded trees" introduced in the seminal work by Childs, Cleve, Deotto, Farhi, Gutmann and Spielman (STOC 2003), there exists a quantum distributed algorithm that transfers messages from the entrance of the graph to the exit with message complexity exponentially smaller than 
    any classical algorithm. Our quantum algorithm is based on the recent "succinct" implementation of quantum walks over the welded trees by Li, Li and Luo (SODA 2024). Our classical lower bound is obtained by ``lifting'' the lower bound from Childs, Cleve, Deotto, Farhi, Gutmann and Spielman (STOC 2003) from query complexity to message complexity.
\end{abstract}
\thispagestyle{empty}
\newpage

\section{Introduction}\label{sec:intro}
\paragraph{Quantum distributed computing.}
While the potential of quantum distributed computing has been investigated in early works \cite{Arfaoui+14,Ben-Or+STOC05,Denchev+08,ElkinKNP14,Gavoille2009,Tani+12}, it is only recently that quantum advantage has been actively studied. 

Most works have focused on the number of rounds of communication needed to solve distributed problems, which is called the \emph{round complexity}.  Two main models are considered when considering round complexity: the CONGEST model, in which each communication channel has small bandwidth, and the LOCAL model, in which each communication channel has unlimited bandwidth. 
In the CONGEST model, Le Gall and Magniez~\cite{LeGall+PODC18} showed how to compute the diameter of a network quadratically faster than any classical distributed algorithm by implementing Grover's algorithm~\cite{Grover96} in the distributed setting. This result led to the development of many quantum distributed algorithms showing a polynomial quantum advantage for many important problems \cite{AVPODC22,CensorHillel+2022,FLNP21,Fraigniaud2024,Hasegawa+PODC24,Izumi+PODC19,Izumi+STACS20,LeGall+2023,LMN23,WYPODC22}. 
In the LOCAL model, Le Gall, Nishimura and Rosmanis \cite{gall2018quantum} introduced a computational problem that can be solved by a quantum algorithm using $O(1)$ rounds of communication and showed that any classical algorithm solving the problem requires $\Omega(n)$ rounds of communication, where~$n$ is the size of the network. Refs.~\cite{A+24,B+24,A+25,CR+STOC24} have recently investigated the existence of quantum advantage for locally-checkable problems in the LOCAL model (the problem considered in \cite{gall2018quantum} is not locally-checkable). In particular, Ref.~\cite{B+24} showed that for any $\Delta\ge 3$ there exists a locally-checkable problem on a graph of maximum degree $\Delta$ that can be solved in $O(1)$ rounds in the quantum LOCAL model but requires $\Omega(\min\{ \Delta,\log_\Delta\log n \})$ rounds in the classical LOCAL model. All those quantum advantages in the LOCAL model, however, are for ``artificial'' problems.

A recent work by Dufoulon, Magniez and Pandurangan \cite{Dufoulon+PODC25} has investigated quantum advantage for another basic complexity measure in distributed computing, the \emph{message complexity}, which corresponds to the overall amount of communication used by the distributed algorithm (i.e., its communication complexity).\footnote{We note that from a purely theoretical perspective, quantum advantage in message complexity were already established by early works in two-party communication complexity (including several examples of exponential advantage \cite{Buhrman+01,RazSTOC99}) and the works \cite{A+24,A+25,LeGall+STACS19} already mentioned in the LOCAL model. The goal of research conducted in \cite{Dufoulon+PODC25} and in the present paper is rather to design quantum distributed algorithms that outperform classical distributed algorithms for concrete problems of interest to the distributed computing community.} 
For leader election in 2-diameter graphs, for instance, Ref.~\cite{Dufoulon+PODC25} presented a quantum algorithm with message complexity $\tilde O(n^{2/3})$, which beats the classical $\tilde \Theta (n)$ bound from \cite{Chatterjee+20}. We also mention a recent work investigating the tradeoff between message complexity and round complexity of quantum distributed algorithms for crash-tolerant consensus~\cite{ICALP24}.

\paragraph{Our result.} 
In this paper, we show that an exponential quantum advantage can be achieved for the message complexity of a fundamental problem.

We consider a network represented by a connected, undirected and unweighted graph, in which nodes represent processors and edges represent communication channels. In the classical setting, the nodes are classical computers and the channels are classical. In the quantum setting, the nodes are quantum computers and the channels allow quantum communication. The nodes have initially no information about the topology of the graph: each node only knows the number of nodes to which it is connected.
Nodes can can only send messages to their neighbors. 
 There are two special nodes labeled $\entr$ and $\exit$.  
The node labeled $\entr$ has as input some data $\mess$ consisting of $\mlength$ bits (all the other nodes have no input). The goal is to send the data $\mess$ from the node labeled $\entr$ to the node labeled $\exit$. We call this problem the \emph{point-to-point routing problem}. 
This problem is difficult since the nodes do not know how to reach $\exit$ from $\entr$. The trivial algorithm will thus broadcast (using a graph traversal algorithm, e.g., depth-first search) the data $\mess$ to all the nodes of the network, which has message complexity $O(b\cdot N)$, where $N$ is the number of nodes. It is easy to see that there exists graphs (e.g., a line) for which this trivial algorithm is optimal. In this paper, we show that there exists a family of graphs in which the point-to-point routing problem can be solved with message complexity $O(b\cdot\poly(\log N))$ in the quantum setting but requires message complexity $\Omega(N)$ in the classical setting.

\begin{figure}[th]
\vspace{7mm}
\centering
\begin{tikzpicture}[scale=0.45,roundnode/.style={circle, draw=black!100, fill=black, very thick, minimum size=1mm,scale=0.5},roundnode2/.style={circle, draw=black!100, fill=blue!90, very thick, minimum size=1mm,scale=0.5},roundnode3/.style={circle, draw=black!100, fill=black!20!green, very thick, minimum size=1mm,scale=0.3}]
\newcommand\XA{2}
\newcommand\YA{0.5}
\node[roundnode,draw=black,minimum height=3mm] (d1) at (0
*\XA,15*\YA) {}; 
\node[draw=none,fill=none,] at (-1.0*\XA,15*\YA) {\small $\entr$};
\draw [thin,black] (0*\XA,15*\YA) -- (1*\XA,7*\YA){};
\draw [thin,black] (0*\XA,15*\YA) -- (1*\XA,23*\YA){};
\foreach \i in {7,23} {
    \node[roundnode,draw=black,minimum height=3mm] (d1) at (1*\XA,\i*\YA) {}; 
    \draw [thin,black] (1*\XA,\i*\YA) -- (2*\XA,\i*\YA+4*\YA){};
    \draw [thin,black] (1*\XA,\i*\YA) -- (2*\XA,\i*\YA-4*\YA){};
}
\foreach \i in {3,11,19,27} {
    \node[roundnode,draw=black,minimum height=3mm] (d1) at (2*\XA,\i*\YA) {}; 
    \draw [thin,black] (2*\XA,\i*\YA) -- (3*\XA,\i*\YA+2*\YA){};
    \draw [thin,black] (2*\XA,\i*\YA) -- (3*\XA,\i*\YA-2*\YA){};
}
\foreach \i in {1,5,9,13,17,21,25,29} {
    \node[roundnode,draw=black,minimum height=3mm] (d1) at (3*\XA,\i*\YA) {}; 
    \draw [thin,black] (3*\XA,\i*\YA) -- (4*\XA,\i*\YA+1*\YA){};
    \draw [thin,black] (3*\XA,\i*\YA) -- (4*\XA,\i*\YA-1*\YA){};
}
\foreach \i in {0,2,4,6,8,10,12,14,16,18,20,22,24,26,28,30} {
    \node[roundnode,draw=black,minimum height=3mm] (d1) at (4*\XA,\i*\YA) {}; 
}

\draw [thin,black] (4*\XA,0*\YA) -- (8*\XA,4*\YA){};
\draw [thin,black] (4*\XA,6*\YA) -- (8*\XA,4*\YA){};
\draw [thin,black] (4*\XA,6*\YA) -- (8*\XA,10*\YA){};
\draw [thin,black] (4*\XA,18*\YA) -- (8*\XA,10*\YA){};
\draw [thin,black] (4*\XA,18*\YA) -- (8*\XA,24*\YA){};
\draw [thin,black] (4*\XA,2*\YA) -- (8*\XA,24*\YA){};
\draw [thin,black] (4*\XA,2*\YA) -- (8*\XA,30*\YA){};
\draw [thin,black] (4*\XA,26*\YA) -- (8*\XA,30*\YA){};
\draw [thin,black] (4*\XA,26*\YA) -- (8*\XA,2*\YA){};
\draw [thin,black] (4*\XA,12*\YA) -- (8*\XA,2*\YA){};
\draw [thin,black] (4*\XA,12*\YA) -- (8*\XA,28*\YA){};
\draw [thin,black] (4*\XA,22*\YA) -- (8*\XA,28*\YA){};
\draw [thin,black] (4*\XA,22*\YA) -- (8*\XA,0*\YA){};
\draw [thin,black] (4*\XA,30*\YA) -- (8*\XA,0*\YA){};
\draw [thin,black] (4*\XA,30*\YA) -- (8*\XA,22*\YA){};
\draw [thin,black] (4*\XA,14*\YA) -- (8*\XA,22*\YA){};
\draw [thin,black] (4*\XA,14*\YA) -- (8*\XA,26*\YA){};
\draw [thin,black] (4*\XA,4*\YA) -- (8*\XA,26*\YA){};
\draw [thin,black] (4*\XA,4*\YA) -- (8*\XA,14*\YA){};
\draw [thin,black] (4*\XA,28*\YA) -- (8*\XA,14*\YA){};
\draw [thin,black] (4*\XA,28*\YA) -- (8*\XA,20*\YA){};
\draw [thin,black] (4*\XA,10*\YA) -- (8*\XA,20*\YA){};
\draw [thin,black] (4*\XA,10*\YA) -- (8*\XA,6*\YA){};
\draw [thin,black] (4*\XA,8*\YA) -- (8*\XA,6*\YA){};
\draw [thin,black] (4*\XA,8*\YA) -- (8*\XA,12*\YA){};
\draw [thin,black] (4*\XA,20*\YA) -- (8*\XA,12*\YA){};
\draw [thin,black] (4*\XA,20*\YA) -- (8*\XA,18*\YA){};
\draw [thin,black] (4*\XA,16*\YA) -- (8*\XA,18*\YA){};
\draw [thin,black] (4*\XA,16*\YA) -- (8*\XA,8*\YA){};
\draw [thin,black] (4*\XA,24*\YA) -- (8*\XA,8*\YA){};
\draw [thin,black] (4*\XA,24*\YA) -- (8*\XA,16*\YA){};
\draw [thin,black] (4*\XA,0*\YA) -- (8*\XA,16*\YA){};

\foreach \i in {0,2,4,6,8,10,12,14,16,18,20,22,24,26,28,30} {
    \node[roundnode,draw=black,minimum height=3mm] (d1) at (8*\XA,\i*\YA) {}; 
}
\foreach \i in {1,5,9,13,17,21,25,29} {
    \node[roundnode,draw=black,minimum height=3mm] (d1) at (9*\XA,\i*\YA) {}; 
    \draw [thin,black] (9*\XA,\i*\YA) -- (8*\XA,\i*\YA+1*\YA){};
    \draw [thin,black] (9*\XA,\i*\YA) -- (8*\XA,\i*\YA-1*\YA){};
}
\foreach \i in {3,11,19,27} {
    \node[roundnode,draw=black,minimum height=3mm] (d1) at (10*\XA,\i*\YA) {}; 
    \draw [thin,black] (10*\XA,\i*\YA) -- (9*\XA,\i*\YA+2*\YA){};
    \draw [thin,black] (10*\XA,\i*\YA) -- (9*\XA,\i*\YA-2*\YA){};
}
\foreach \i in {7,23} {
    \node[roundnode,draw=black,minimum height=3mm] (d1) at (11*\XA,\i*\YA) {}; 
    \draw [thin,black] (11*\XA,\i*\YA) -- (10*\XA,\i*\YA+4*\YA){};
    \draw [thin,black] (11*\XA,\i*\YA) -- (10*\XA,\i*\YA-4*\YA){};
}
\node[roundnode,draw=black,minimum height=3mm] (d1) at (12
*\XA,15*\YA) {}; 
\node[draw=none,fill=none,] at (13.0*\XA,15*\YA) {\small $\exit$};
\draw [thin,black] (12*\XA,15*\YA) -- (11*\XA,7*\YA){};
\draw [thin,black] (12*\XA,15*\YA) -- (11*\XA,23*\YA){};

\draw [thick,black,<->] (-0.3*\XA,-2*\YA) -- (4.3*\XA,-2*\YA){};
\node[draw=none,fill=none,] at (2*\XA,-3.5*\YA) {height $n$};
\draw [thick,black,<->] (7.7*\XA,-2*\YA) -- (12.3*\XA,-2*\YA){};
\node[draw=none,fill=none,] at (10*\XA,-3.5*\YA) {height $n$};
\end{tikzpicture}
\caption{An instance of welded trees in $\Gg_{n}$ for $n=4$.}\label{fig1}
\end{figure}
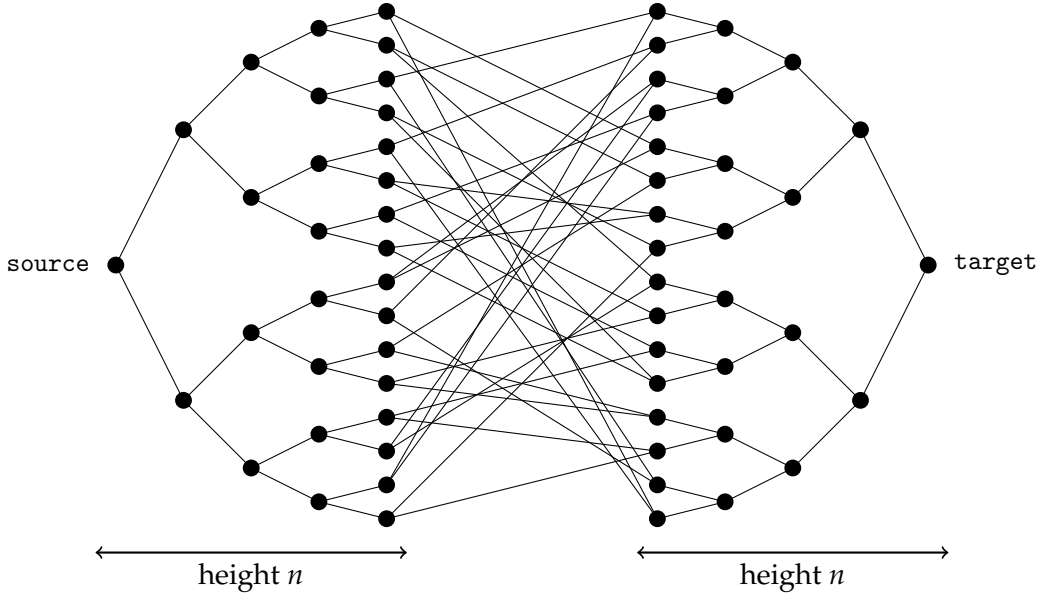


The family of graphs we consider is the family of welded trees introduced in \cite{Childs+STOC03} to show an exponential advantage of quantum-walks-based algorithms and then further studied in, e.g., \cite{Altia+03,Belovs24,Fenner+03,Jeffery+STOC23,Li+24}. This family, indexed by a positive integer $n$, is denoted by $\Gg_n$.
An instance of welded trees in $\Gg_n$ for $n=4$ is represented in \cref{fig1}. A graph in $\Gg_n$ has $2^{n+2}-2$ nodes. It consists of two balanced binary trees of height $n$, whose roots are labeled $\entr$ and $\exit$, respectively, connected together by an arbitrary cycle that alternates between the leaves of the two trees such that every leaf on the left is connected to two leaves on the right (and vice versa).\footnote{Such a cycle can be constructed by choosing a leaf on the left and connecting it to a leaf on the right. The latter is then connected to a leaf on the left chosen among the remaining ones. This process continues, alternating sides, until obtaining a cycle in which every leaf on the left is connected to two leaves on the right (and vice versa).}
Here are our main results:
\begin{theorem}\label{main_thm1}
There exists a quantum distributed algorithm that solves with probability at least $1-1/\exp(n)$ the point-to-point routing problem in $\Gg_n$ with message complexity $O(b\cdot\poly(n))$. This upper bound holds even when $\mess$ is a quantum data of $b$ qubits if $\entr$ has several copies of it.
\end{theorem}
\begin{theorem}\label{main_thm2}
For any $q\in[2^{-b},1]$,
any classical distributed algorithm that solves the point-to-point routing problem in $\Gg_n$ with probability at least $q$ has message complexity $\Omega((q-\frac{1}{2^b})^{1/5}\cdot 2^{n/20})$. This lower bound holds even if the nodes share prior randomness. 
\end{theorem}
For $b=\poly(n)$, Theorems \ref{main_thm1} and \ref{main_thm2} imply an exponential separation between the quantum and classical communication complexities.
While this separation holds only for a very specific network topology (the graphs in $\Gg_n$), this nevertheless exhibits the potential of quantum communication for a new kind of task in distributed computing.

\paragraph{Technical overview of the quantum upper bound (\cref{main_thm1}).}
The goal is to produce a quantum algorithm that solves the problem with polynomial message complexity. This is done by combining three ideas: (i) there exists a quantum walk that is defined on any graph and that traverses the welded trees in polynomial time step; (ii) it is possible to reproduce this quantum walk in the distributed quantum setting and (iii) reproducing this quantum walk has polynomial message complexity.

As already mentioned, there exist several types of quantum walks known to traverse the welded trees in polynomial time \cite{Altia+03,Belovs24,Childs+STOC03,Jeffery+STOC23,Li+24}. Most of them however seem difficult to implement in the distributed setting. First of all, implementing directly continuous-time quantum walks \cite{Altia+03,Childs+STOC03,Fenner+03} does not seem feasible since the distributed setting is intrinsically discrete. While continous-time quantum walks can be generically converted to discrete-time walks \cite{Childs10}, this discretization only leads to a quantum circuit implementing the walk in a discrete number of steps. It is unclear how to efficiently implement this circuit in the distributed setting, where each node can only interact with its neighbors and perform local operations. A similar difficulty arises when trying to implement the recent discrete-time quantum walks for the welded trees by Belovs \cite{Belovs24} and Jeffery and Zur \cite{Jeffery+STOC23}. Fortunately, the recent ``succinct'' discrete-time quantum walk by Li, Li and Luo \cite{Li+24} is more suited to a distributed implementation.

Implementing the discrete-time quantum walk from \cite{Li+24} in the distributed setting with polynomial message complexity nevertheless requires new insights. Let us first explain the approach from~\cite{Li+24} (more details are given in \cref{rw}). Let $G=(V,E)$ be a welded tree. The quantum walk is a discrete-time quantum process defined
over the vector space
\[
\mathtt{Span}\left\{ \ket{u}\ket{v}\mid u\in V,\, v \in \neighbor{u}\right\}
\]
in which the state of the system at Step $t$ is of the form 
\begin{equation}\label{eq:istate}
    \ket{\Phi(t)} = \sum_{u\in V}\sum_{v \in \neighbor{u}} \Phi_{u,v}(t)\ket{u}\ket{v},
\end{equation}
where $\neighbor{u}$ denotes the set of neighbors of $u$ (similarly, we use $\deg(u)$ to denote the degree of $u$).
One step of the walk is implemented by applying
the evolution operator $U = SC$, where $S$ is the operator that swaps the two registers, and 
\[
    C = \sum_{u\in V} \ket{u}\bra{u}\otimes G_u
\]
is the coin operator (here $G_u$ is a Grover diffusion operator acting on the second register). 

Our key idea to implement this walk in the distributed setting and send the data $\mess$ from $\entr$ to $\exit$ is to replace the state of \cref{eq:istate} by 
\[
    \sum_{u\in V}\sum_{v \in \neighbor{u}} \Phi_{u,v}(t)\ket{\mathcal{R}_{u\gets v}},
\]
where $\ket{\mathcal{R}_{u\gets v}}$ is a state (formally defined in \cref{sec:ub}) of the form 
\begin{equation}\label{eq:pir}
    \ket{\mathcal{R}_{u\gets v}} = \underbrace{\ket{\bot}\cdots\ket{\bot}\ket{\mess}\ket{\bot}\cdots\ket{\bot}}_{\deg(u) \textrm{ registers}}
\end{equation}
in which the position of $\ket{\mess}$ corresponds to the port number of $u$ that links to $v$ (the notion of port number is defined in \cref{sub:dc}). Here $\ket{\bot}$ is a ``dummy'' quantum state of the same dimension as $\ket{x}$ indicating the absence of message. The key insight is that we can implement the coin operator used in \cite{Li+24} by permuting the $\deg(u)$ registers of $\ket{\mathcal{R}_{u\gets v}}$ and adding a phase, and that such an implementation can be done mostly locally, by simply sending one message per iteration of the quantum walk.
The number of messages is thus linear in the number of iterations of the original quantum walk, giving $\poly(n)$ messages in total. Each message consists of one of the registers of \cref{eq:pir}, which leads to an overall message complexity of $O(b\cdot \poly(n))$ qubits.
In fact, our implementation allows us to even send a quantum state $\ket{\varphi}$ from $\entr$ to $\exit$ simply by replacing the classical data $\ket{x}$ by $\ket{\varphi}$ in the implementation.


\paragraph{Technical overview of the classical lower bound (\cref{main_thm2}).}
The proof of the exponential classical complexity follows a path similar to that of~\cite{Childs+STOC03}. However, the distributed setting brings many more challenges. Indeed, having agents that can each try independently to help connect $\entr$ to $\exit$ brings more versatility to the algorithm. 

In~\cite{Childs+STOC03}, initially, the player only knows the name of $\entr$, which leaves it no choice but to either guess the names of the other nodes, which is unlikely as the possibilities exponentially outnumber the nodes, or to start a single exploration from $\entr$ in hope of finding $\exit$. In the distributed setting, every node is aware of itself and does not need to be discovered by a centralized entity before starting to explore the graph, which in this case means starting to propagate messages so that the network collectively learns its own structure. For example, it could be that a node in the middle of the network starts an exploration that links $\entr$ to $\exit$. In fact, any number of nodes can start an exploration in parallel from the first round.

Here is the outline of our proof, accounting for this greater power of the distributed algorithm. First, a way for a centralized player to simulate a distributed algorithm for the point-to-point routing problem is exhibited. Similarly to the centralized problem of~\cite{Childs+STOC03}, the player plays through an oracle that, given an identifier and a port number, returns the corresponding neighbor in the graph. But in this version of the centralized problem, instead of only that of $\entr$, the player knows every identifier in the network, it only ignores how nodes neighbor one another. In this way, one query to the oracle exactly accounts for one message being forwarded from the input of the oracle to the output, and the player learns of the edge between the nodes as the nodes themselves would thanks to the message. Moreover, every identifier being known by the player from the beginning simulates the fact that any node can independently start messaging from the beginning.

As in~\cite{Childs+STOC03}, an equivalence between an exploration of the graph through the use of the oracle and the random embedding of a tree in the graph is established. Still, the possibility for several independent start points of exploration further complicates things. Instead of embedding one tree rooted in $\entr$, the player embeds a number of trees, only bounded by the number of messages in the distributed algorithm and rooted anywhere in the graph. For one, some care is needed in showing the equivalence to embedding the trees in order to bring back some parallelization: the choice for the trees and their embeddings are totally independent while the centralized player, sequentially calling the oracle, could totally condition its different explorations. But the main difficulty is the following one. It is no longer sufficient to show that a tree rooted in the left part of the network cannot reach the rightmost quarter of columns as other trees rooted in the right part could help bridge the gap. The unlikelihood to be shown is then that any pair of trees has low probability of intersecting each other. As a tree has overwhelming chance of being rooted close to the middle of the network, the same way a node independently deciding to message has overwhelming chance of turning out to be close to the middle, it is still true that reaching the left- and right-most quarters of columns is unlikely. From this point it remains to be shown that trees intersecting close to the middle is also unlikely. This time the argument has to be vertical instead of horizontal, i.e., it is no longer a matter of trees reaching the same column, but a matter of trees reaching the same node, knowing that they reach the same column. It is thus shown that a random embedding of a tree infers that any of its nodes has equal probability of being embedded to any node inside a column. The columns close to the middle being of huge size, any two nodes uniformly drawn in one should not collide.

\section{Preliminaries}\label{sec:prelim}
In this paper all graphs are undirected and unweighted. The edges are denoted as $\{u,v\}$, but we will also use the notation $(u,v)$ when we need to order the two extremities. As already mentioned in the introduction, for a graph $G=(V,E)$ and a vertex $u\in G$, we denote by $\deg_G(u)$ its degree in $G$ and by $\mathcal{N}_G(u)$ the set of its neighbors in $G$ (we often omit the subscript $G$ when there is no possibility of confusion). For any positive integer $r$, we write $[r]=\{1,\ldots,r\}$.

We discuss below notions of distributed computing and quantum walks needed for this work.

\subsection{Distributed computing}\label{sub:dc}
We use the model introduced in \cref{sec:intro}, which is standard in distributed computing. In this subsection we first give a few additional technical details, and then define formally the message complexity of classical and quantum distributed algorithms.

The topology of the network is represented by a graph $G = (V, E)$.
All links and nodes (corresponding to the edges and vertices of $G$, respectively) are reliable and suffer no faults. Each node has a distinct identifier from a domain~$I$ of size $\poly(|V|)$. We denote by $\mathsf{ID}\subseteq I$ the set of all identifiers of the nodes. As specified in \cref{sec:intro}, initially each node knows nothing about the topology of the network except its identifier and its degree. The only exception is $\entr$ and $\exit$, who know that they are, respectively, $\entr$ and $\exit$ (but nothing else).
We stress that nodes initially do not even know their neighbors: we are actually working in the so-called port-numbering model in which for each node $u$ there exists a bijection $p_u\colon\{1,\ldots,\deg(u)\}\to\{v\in V\:|\:\{u,v\}\in E\}$ that is used when sending messages (see details below).

Executions proceed with round-based synchrony: in each round, each node performs local computation based on information it has, prepares messages to be transmitted to its neighbors and send them using the port numbers (we explain below in more details how this is done). As usual in distributed computing, the whole algorithm stops when one of the nodes of the network decides to stop. For instance, for the point-to-point routing problem, our quantum algorithm will stop when the node labeled $\exit$ outputs ``success''. 

Since our goal is to prove a separation between quantum and classical distributed algorithms, in the classical randomized setting we allow the nodes to share prior randomness (our classical lower bound works for this case), while in the quantum setting we do not assume that the nodes have any prior shared randomness or shared entanglement (our quantum upper bound works for this case). Note that we do not assume that the nodes know the number of nodes of the graph (but the classical lower bound of \cref{main_thm2} will actually apply to the case where the nodes know $\abs{V}$ as well). 

\paragraph{Classical deterministic message complexity.}
At the end of each round, node $u$ chooses one message $m_{u\to i}\in\{0,1\}^{\ast}$ for each $i\in [\deg(u)]$. 
 We use $\perp$ to denote the empty message. The message $m_{u\to i}$ is then sent to node $p_u(i)$ if $m_{u\to i}\neq\perp$ . Nothing is sent from $u$ to $p_u(i)$ if $m_{u\to i} = \perp$. The classical deterministic message complexity of a round is the sum, over all nodes, of the number of bits sent at (the end of) this round. The message complexity of the protocol is the sum of the message complexities of all rounds. For clarity, we will very often write $m_{u\to v}$ instead of $m_{u\to i}$, with $v=p_u(i)$. However, we stress that $u$ does not explicitly know $v$.



\paragraph{Deterministic configuration of a classical system.} In order to define the quantum message complexity, we will need an even more precise definition of the sending operation. We assume that each node $u$ has $2\deg(u)+1$ registers: a local register $\mathsf{D}_u$ in which the local data (including the node's identifier and possible input) and local memory are stored, an emission register $\mathsf{E}_{u\to v}$ for each $v\in\neighbor{u}$, and a reception register $\mathsf{R}_{u\leftarrow v}$ for each $v\in\neighbor{u}$.\footnote{Again, since $u$ does not know the bijection $p_u$, formally we should rather write $\mathsf{E}_{u\to i}$ and $\mathsf{R}_{u\gets i}$, with $v=p_u(i)$. For clarity we nevertheless prefer to use $\mathsf{E}_{u\to v}$ and $\mathsf{R}_{u\gets v}$.}
The configuration of the system is the concatenation of all those registers for all the nodes of the network. 

Just before sending the message, the contents of $\mathsf{E}_{u\to v}$ is the (possibly empty) message $m_{u\to v}$,
while the contents of $\mathsf{R}_{v\gets u}$ is $\perp$ for preparing for the potential reception. Just after sending the messages (at the very end of the round), the contents of $\mathsf{E}_{u\to v}$ becomes $\perp$ since the message has been delivered, while the contents of $\mathsf{R}_{v\gets u}$ becomes $m_{u\to v}$. This means that sending a message corresponds to swapping the contents of Registers $\mathsf{E}_{u\to v}$ and $\mathsf{R}_{v\gets u}$, for all pairs $(u,v)\in E$.
With these definitions, the message complexity of one round, which we also refer to as the message complexity of the corresponding configuration, is the sum of the sizes of the contents of all emission registers that do not contain $\perp$ just before sending the messages. Note that since registers containing $\perp$ are not counted, empty messages have no impact on the message complexity (which agrees with the interpretation in which empty messages were not sent).

\paragraph{Classical randomized message complexity.}
In the classical randomized setting, each node can access an infinite sequence of local random bits and the nodes can initially share randomness. Once the randomness has been fixed, the algorithm becomes deterministic. The message complexity is defined as the maximum over all random bits (including the prior shared randomness) of the message complexity of the corresponding deterministic algorithm. 

Note that the system’s conﬁguration of a classical randomized algorithm becomes a probability distribution over every possible deterministic conﬁguration. The transition from one conﬁguration to another is done according to a stochastic transformation, made of two steps: (1) Perform the same local stochastic operation to each node on their registers (the operation depends on the local random bits); (2) Send the messages prepared in Step (1) by swapping the contents of Registers $\mathsf{E}_{u\to v}$ and $\mathsf{R}_{v\gets u}$, for all pairs $(u,v)\in E$.


\paragraph{Quantum message complexity.}
In the quantum setting, messages now consist of qubits.  As in \cite[Appendix~A]{Dufoulon+PODC25}, we implement quantum distributed algorithms by replacing $\mathsf{D}_u$, $\mathsf{E}_{u\to v}$ and $\mathsf{R}_{v\to u}$ by quantum registers.\footnote{A slight difference is that in Ref.~\cite{Dufoulon+PODC25} messages are sent at the beginning of a round. In our case, we assume that that the messages are sent at the end of a round since this will be more convenient to describe our quantum algorithm in \cref{sec:ub}.} This means that at every time, the system’s conﬁguration is a quantum superposition of every possible deterministic conﬁguration (we use a special quantum state $\ket{\bot}$, orthogonal to all other possible quantum messages, to encode the empty message). The transition from one conﬁguration to another is done according to a unitary transformation, made of two steps: (1) Perform the same local unitary to each node on their registers; (2) Send the messages prepared in Step (1) by swapping the contents of Registers $\mathsf{E}_{u\to v}$ and $\mathsf{R}_{v\gets u}$, for all pairs $(u,v)\in E$. 

As in \cite{Dufoulon+PODC25}, we define the quantum message complexity of one round as the maximum message complexity of a deterministic classical conﬁguration appearing with non-zero amplitude, and the overall quantum message complexity as the sum of the complexities of each round.

\subsection{Quantum walks}\label{rw}
The proof of the quantum upper bound relies on a quantum walk called the Grover walk, which can be defined on any undirected graph. Its definition uses the following $d\times d$ unitary operator called the Grover operator:
$$
G_d = \left(\begin{matrix}
   \frac{2}{d}-1 & \frac{2}{d} & \ldots & \frac{2}{d} \\
   \frac{2}{d} & \ddots & \ddots & \vdots \\
   \vdots & \ddots & \ddots & \frac{2}{d} \\
   \frac{2}{d} & \ldots & \frac{2}{d} & \frac{2}{d}-1 \\
\end{matrix}\right)
= 2\ket{s}\bra{s} - I_d,
$$
where 
$
\ket{s} = \frac{1}{\sqrt{d}}\sum_i \ket{i}.
$
In particular we have 
\begin{equation}\label{eq:G}
    G_2 = \left(\begin{matrix}
   0 & 1 \\
   1 & 0
\end{matrix}\right),\hspace{2mm}
G_3 = \frac{1}{3}\left(\begin{matrix}
   -1 & 2 & 2 \\
   2 & -1 & 2\\
  2 & 2 & -1
\end{matrix}\right).
\end{equation}

\begin{definition}\label{def:QW}
	The Grover walk on an undirected graph $G=(V,E)$ is a discrete-time quantum process defined 
    over the vector space
    \[
    \mathtt{Span}\left\{ \ket{u}\ket{v}\mid u\in V,\, v \in \neighbor{u}\right\} 
    \]
    by an initial state $\ket{\Phi(0)}$ of the form 
    \[
     \ket{\Phi(0)} = \sum_{u\in V}\sum_{v \in \neighbor{u}} \Phi_{u,v}(0)\ket{u}\ket{v},
     \]
    and an evolution operator $U = SC$, with
    \begin{align*}
    S &= \sum_{u\in V}\sum_{v \in \neighbor{u}} \ket{v,u}\bra{u,v},\\
    C &= \sum_{u\in V} \ket{u}\bra{u}\otimes G_u
    ,
    \end{align*}
    where each $G_u$ acts as $G_{\deg(u)}$ on the subspace 
    $\mathtt{Span}\left\{\ket{v}\mid v\in\neighbor{u}\right\}$, i.e., 
    \begin{equation}\label{eq:coin}
    C=\sum_{u\in V}\sum_{v\in \neighbor{u}}\left(\frac{2}{\deg(u)}\sum_{w\in \neighbor{u}}\ket{u,w}\bra{u,v}\right) - \ket{u,v}\bra{u,v}.
    \end{equation}
\end{definition}

Li, Li and Luo \cite{Li+24} have shown that this specific quantum walk traverses the welded trees graph with inverse-polynomial probability in a polynomial number of time steps. Here is the formal statement: 
\begin{proposition}[\cite{Li+24}]\label{lemma:QWcross}        
       Consider a welded trees graph $G$ in $\Gg_n$ and the Grover walk on $G$ with initial state 
    \begin{equation}\label{eq:init}
       \ket{\Phi(0)} = \frac{1}{\sqrt{2}}\sum_{v \in \neighbor{\entr}}\ket{\entr}\ket{v}.
    \end{equation}
    Let
     \begin{equation}\label{eq:t}
     \ket{\Phi(t)} =  U^{t}\ket{\Phi(0)}= \sum_{u\in V}\sum_{v \in \neighbor{u}} \Phi_{u,v}(t)\ket{u}\ket{v}
     \end{equation}
	denote the state of the system at time $t\ge 1$ and define
	$$
	p(t) = \sum_{v \in \neighbor{\exit}}|\braket{\exit, v\mid \Phi(t)}|^2,
	$$
	which represents the probability of measuring the $\exit$ at Step $t$.
	There exists $\hat T \in [2n, 3.6 n \log n]$ such that
	$$
	p(\hat T) > \frac{1}{20n}.
	$$
\end{proposition}

 We analyze in more detail the relation between $\ket{\Phi(t-1)}$ and $\ket{\Phi(t)}$. First, using \cref{eq:coin}, we have
\begin{align*}
 C\ket{\Phi(t-1)}&=
 \sum_{u\in V}\sum_{v\in \neighbor{u}}\Phi_{u,v}(t-1)\left[\left(\frac{2}{\deg(u)}\sum_{w\in \neighbor{u}}\ket{u,w}\right) - \ket{u,v}\right]\\
  &= \sum_{u\in V}\sum_{v\in \neighbor{u}}\sum_{w\in \neighbor{u}}\frac{2}{\deg(u)}\Phi_{u,v}(t-1)\ket{u,w} - \sum_{u\in V}\sum_{v\in \neighbor{u}}\Phi_{u,v}(t-1)\ket{u,v}\\
  &= \sum_{u\in V}\sum_{v\in \neighbor{u}}\sum_{w\in \neighbor{u}}\frac{2}{\deg(u)}\Phi_{u,w}(t-1)\ket{u,v} - \sum_{u\in V}\sum_{v\in \neighbor{u}}\Phi_{u,v}(t-1)\ket{u,v}\\
  &= \sum_{u\in V}\sum_{v\in \neighbor{u}}\left[\frac{2}{\deg(u)}\left(\sum_{w\in \neighbor{u}}\Phi_{u,w}(t-1)\right) - \Phi_{u,v}(t-1)\right]\ket{u,v}.\\
\end{align*}
We thus obtain 
\[
\ket{\Phi(t)}=SC\ket{\Phi(t-1)}=
 \sum_{u\in V}\sum_{v\in \neighbor{u}}\left[\frac{2}{\deg(v)}\left(\sum_{w\in \neighbor{v}}\Phi_{v,w}(t-1)\right) - \Phi_{v,u}(t-1)\right]\ket{u,v},
\]
which implies the following relations between $\Phi_{u,v}(t)$ and $\Phi_{u,v}(t-1)$:
    \begin{equation}\label{eq:rec}
        \forall (u,v) \in E, \quad \Phi_{u,v}(t) = \frac{2}{\deg v}\left(\sum_{w\in \neighbor{v}}\Phi_{v,w}(t-1)\right) - \Phi_{v,u}(t-1).
    \end{equation}

\section{Quantum Upper Bound}\label{sec:ub}
This section focuses on proving the quantum upper bound stated in Theorem~\ref{main_thm1}: the point-to-point routing problem over $\Gg_n$ has message complexity $\poly(n)$ in the quantum setting. We will actually show more generally how to send \emph{quantum} data from $\entr$ to $\exit$. In this section, we thus replace the $b$-bit classical data $\mess$ by an arbitrary $b$-qubit (unknown) quantum state $\ket{\qmess}$.\footnote{For simplicity we assume that $\ket{\qmess}$ is a pure state but our algorithm also works for mixed states.} Our final algorithm will work  if $\entr$ has a polynomial number of copies of $\ket{\qmess}$ (or is able to create these copies by itself, which is obviously the case if the data is classical). 

\paragraph{Definitions.}

We first need to  define a quantum state $\ket{\perp}$ that is orthogonal to all possible data~$\ket{\qmess}$. For this purpose, we embed $\ket{\qmess}$ into the vector space $\Hh=\Comp^{2^{b+1}}$ as $\ket{\tilde\qmess}=\ket{\qmess}\ket{0}$ and define
\[
\ket{\perp}=\ket{0}^{\otimes b}\ket{1},
\]
which is clearly orthogonal to all possible states $\ket{\tilde \qmess}$. We denote by $\Hh'$ the subspace of $\Hh$ corresponding to all valid data, i.e., 
$\Hh'=\{\ket{\tilde \qmess}\mid \ket{\qmess}\in\Comp^{2^b}\}$. We consider the measurement $\Pi=\{\Pi_0=\ket{\perp}\bra{\perp}, \Pi_1=I-\Pi_0\}$ over $\Hh$.

Our algorithm will use two unitary operators $\tilde G_2$ and $\tilde G_3$ closely associated to the operators $G_2$ and $G_3$ introduced in \cref{eq:G}.
The operator $\tilde G_2$ is defined as follows over two registers $\Rr_1,\Rr_2$ each containing a state in $\Hh$:  
for any state $\psi\in\Hh'$,
\begin{align*}
\tilde G_2\colon \ket{\psi}_{\Rr_1}\ket{\perp}_{\Rr_2}\mapsto \ket{\perp}_{\Rr_1}\ket{\psi}_{\Rr_2},\\
\tilde G_2\colon \ket{\perp}_{\Rr_1}\ket{\psi}_{\Rr_2}\mapsto \ket{\psi}_{\Rr_1}\ket{\perp}_{\Rr_2}.
\end{align*}
The action of $\tilde G_2$ on states that are not of the above form is defined arbitrary. 
Note that $\tilde G_2$ can be implemented straightforwardly by a $\textsf{SWAP}$ operator.
The operator $\tilde G_3$ is defined as follows over three registers $\Rr_1,\Rr_2, \Rr_3$ each containing a state in $\Hh$: 
for any state $\ket{\psi}\in\Hh'$,
\begin{align*}
\tilde G_3\colon \ket{\psi}_{\Rr_1}\ket{\perp}_{\Rr_2}\ket{\perp}_{\Rr_3}&\mapsto 
-\frac{1}{3}\ket{\psi}_{\Rr_1}\ket{\perp}_{\Rr_2}\ket{\perp}_{\Rr_3} + 
\frac{2}{3}\ket{\perp}_{\Rr_1}\ket{\psi}_{\Rr_2}\ket{\perp}_{\Rr_3}+
\frac{2}{3}\ket{\perp}_{\Rr_1}\ket{\perp}_{\Rr_2}\ket{\psi}_{\Rr_3},\\
\tilde G_3\colon \ket{\perp}_{\Rr_1}\ket{\psi}_{\Rr_2}\ket{\perp}_{\Rr_3}&\mapsto 
\phantom{-}\frac{2}{3}\ket{\psi}_{\Rr_1}\ket{\perp}_{\Rr_2}\ket{\perp}_{\Rr_3} - 
\frac{1}{3}\ket{\perp}_{\Rr_1}\ket{\psi}_{\Rr_2}\ket{\perp}_{\Rr_3}+
\frac{2}{3}\ket{\perp}_{\Rr_1}\ket{\perp}_{\Rr_2}\ket{\psi}_{\Rr_3},\\
\tilde G_3\colon \ket{\perp}_{\Rr_1}\ket{\perp}_{\Rr_2}\ket{\psi}_{\Rr_3}&\mapsto 
\phantom{-}\frac{2}{3}\ket{\psi}_{\Rr_1}\ket{\perp}_{\Rr_2}\ket{\perp}_{\Rr_3} + 
\frac{2}{3}\ket{\perp}_{\Rr_1}\ket{\psi}_{\Rr_2}\ket{\perp}_{\Rr_3}-
\frac{1}{3}\ket{\perp}_{\Rr_1}\ket{\perp}_{\Rr_2}\ket{\psi}_{\Rr_3}.
\end{align*}
Again, the action of $\tilde G_3$ on states that are not of the above form is defined arbitrary.

\paragraph{Basic algorithm.}
While the Grover walk allows us to traverse the welded trees graph efficiently (\cref{lemma:QWcross}), it does not follow the distributed setting.
Our goal is to reproduce its dynamic in a quantum distributed setting. 
Recall that in the distributed setting, each node $u$ has the registers $\mathsf{D}_u$ and $\mathsf{E}_{u\to v}$, $\mathsf{R}_{u\gets v}$ for all $v\in \neighbor{u}$. All registers $\mathsf{E}_{u\to v}$, $\mathsf{R}_{u\gets v}$ will contain quantum states in $\Hh$.
Communication occurs by swapping registers $\mathsf{E}_{u\to v}$ and $\mathsf{R}_{v\gets u}$ for all $(u,v) \in E$. 
Local operations must be the same for all nodes (except possibly for  $\entr$ and $\exit$). 

Our basic distributed quantum algorithm for traversing the welded trees graph is the algorithm $\Aa(T)$ described below. This algorithm only requires one copy of the quantum data. Initially, all the registers $\mathsf{E}_{u\to v}$, $\mathsf{R}_{u\gets v}$ are initialized to $\ket{\perp}$.

\begin{center}
		\begin{minipage}{16.6 cm} \vspace{2mm}
		\begin{algorithm}[H]
			\SetAlgorithmName{}{}{List of Algorithms}
			\nonl \hspace{-4mm}Algorithm $\Aa(T)$ 
			\tcp*[l]{Distributed implementation of the quantum walk}
            \nonl \textbf{Input:} $\entr$ receives one copy of the data $\ket{\varphi} \in \mathbb C^{2^b}$


            \nonl \textbf{Output:} Node $\exit$ outputs either ``success'' or ``failure''\vspace{2mm}

             $\entr$ prepares the state 
            \[
                \frac{1}{\sqrt{2}} \left( \ket{\tilde \qmess}_{\mathsf{R}_{\entr\gets v_1}}\ket{\bot}_{\mathsf{R}_{\entr\gets v_2}} + \ket{\bot}_{\mathsf{R}_{\entr\gets v_1}}\ket{\tilde \qmess}_{\mathsf{R}_{\entr\gets v_2}}\right)
            \]
            on Registers $(\mathsf{R}_{\entr\gets v_1},\mathsf{R}_{\entr\gets v_2})$, where $\{v_1,v_2\}=\neighbor{\entr}$\label{step:init}


            \For{$t$ {\bf from }$1$ \KwTo $T$\label{step:forbegin}}
					{

                    Each node $u\in V$ applies the operator $\tilde G_{\deg(u)}$ on the set of registers $(\mathsf{R}_{u\gets v})_{v\in\neighbor{u}}$\label{step:diff}

                    Each node $u\in V$ applies, for each $v \in \neighbor{u}$, the operator  $\textsf{SWAP}$ on  $(\mathsf{E}_{u\to v},\mathsf{R}_{u\gets v})$\label{step:swap}

                    For each $(u,v)\in E$, $\textsf{SWAP}$ is applied on $(\mathsf{E}_{u\to v},\mathsf{R}_{v\gets u})$\hspace{4mm}\tcp*[l]{communication}\label{step:communication}

					}


            $\exit$ applies the measurement $\Pi$ on each of $\mathsf{R}_{\exit \gets v}$ for $v\in \mathcal{N(\exit)}$ \label{step:measure}

            \If{\rm one of the outcomes is 1}
				{{\bf target returns} ``success'' \hspace{0mm}\tcp*[l]{the register with outcome 1 contains $\ket{\varphi}$}}
			\lElse
				{{\bf target returns} ``failure''}
			\end{algorithm}
		\end{minipage}
	\end{center}
	
We now analyze the success probability and the message complexity of this algorithm.

\begin{lemma}\label{lemma:distrib_cross}
    The message complexity of $\Aa(T)$ is $O(b\cdot T)$ for a message of length $b$ qubits.
	Furthermore, there exists a value $\hat T\in [2n, 3.6 n \log n]$ such that, after running $\Aa(\hat T)$,
    \[
    \mathbb P(\exit \text{ outputs ``success''} ) > \frac{1}{20n}.
    \]
\end{lemma}
\begin{proof}
	Let us denote by $\mathcal E$ the quantum register corresponding to the set of all the emission registers $\mathsf{E}_{u\to v}$, for all $(u,v)\in E$, and by $\mathcal R$ the quantum register corresponding to the set of all reception registers $\mathsf{R}_{u\gets v}$, for all $(u,v)\in E$.
    For each $(u,v)\in E$,
    we define the quantum states 
    \begin{align*}
    \ket{\mathcal{E}_{u\to v}}_{\mathcal{E}} =
    \bigotimes_{u'\in V}\bigotimes_{v' \in \neighbor{u'}} 
    \ket{\psi_{u,v,u',v'}}_{\mathsf{E}_{u'\to v'}} \\
    \ket{\mathcal{R}_{u\gets v}}_{\mathcal{R}} =
    \bigotimes_{u'\in V}\bigotimes_{v' \in \neighbor{u'}} \ket{\psi_{u,v,u',v'}}_{\mathsf{R}_{u'\gets v'}} 
    \end{align*}
    where
    \[
    \ket{\psi_{u,v,u',v'}}=
    \begin{cases}
        \ket{\tilde \qmess}&\textrm{ if }(u,v)=(u',v'),\\
        \ket{\perp}&\textrm{ otherwise }.
    \end{cases}
    \]

    At the end of Step \ref{step:init}, the state of the whole system is 
    \[
	    \ket{\Psi(0)} =\left(\frac{1}{\sqrt 2}\sum_{v \in \neighbor{\entr}} \ket{\mathcal{R}_{\entr\gets v}}_{\mathcal R}\right) \otimes \ket{\bot}^{\otimes 2|E|}_{\mathcal E}.
	\]
    (Hereafter we omit all the registers $\mathsf{D}_u$ since the algorithm will not act on them.)
    Note that this state is analog to the state of \cref{eq:init}, with $\ket{\entr}\ket{v}$ replaced by $\ket{\mathcal{R}_{\entr\gets v}}$.
    
	For any $t\ge 1$, consider the state  of the network after $t$ iterations of the loop of Steps \ref{step:forbegin}--\ref{step:communication}
    . We will show that this state is
	\begin{equation}\label{eq:state}
	\ket{\Psi(t)} =\left(\sum_{u\in V}\sum_{v \in \neighbor{u}} \Phi_{u,v}(t)\ket{\mathcal{R}_{u\gets v}}_{\mathcal R}\right) \otimes \ket{\bot}^{\otimes 2|E|}_{\mathcal E}
    \end{equation}
    at the end of the loop (i.e., after Step \ref{step:communication}), where $\Phi_{u,v}(t)$ are the amplitudes appearing in \cref{eq:t}. In other words, $\ket{\Psi(t)}$ is analog to the state $\ket{\Phi(t)}$ of \cref{eq:t}, with $\ket{u}\ket{v}$ replaced by $\ket{\mathcal{R}_{u\gets v}}$.

    Assume that the state of the system when entering the loop for some $t\ge 1$ is $\ket{\Psi(t-1)}$. 
    After Step \ref{step:diff}, 
    the states becomes 
    \[
    \ket{\Psi'(t-1)} =\left(\sum_{u\in V}\sum_{v \in \neighbor{u}} \Phi_{u,v}(t-1)(\tilde G_{\deg(u)} \otimes I)\ket{\mathcal{R}_{u\gets v}}_{\mathcal R}\right) \otimes \ket{\bot}^{\otimes |E|}_{\mathcal E},
    \]
    where $\tilde G_{\deg(u)} \otimes I$ means $\tilde G_{\deg(u)}$ applied 
    to the set of registers $(\mathsf{R}_{u\gets v})_{v\in\neighbor{u}}$. 
    Since 
    \[
        \tilde G_{\deg(u)} \otimes I =
        \sum_{u\in V}\sum_{v\in \neighbor{u}}\left(\frac{2}{\deg(u)}\sum_{w\in \neighbor{u}}\ket{\mathcal{R}_{u\gets w}}_{\mathcal R}\bra{\mathcal{R}_{u\gets v}}_{\mathcal R}\right) - \ket{\mathcal{R}_{u\gets v}}_{\mathcal R}\bra{\mathcal{R}_{u\gets v}}_{\mathcal R},
    \]
    we obtain
    \begin{align*}
        \ket{\Psi'(t-1)} 
        &= \left(\sum_{u\in V}\sum_{v \in \neighbor{u}} \Phi_{u,v}(t-1)\left[\frac{2}{\deg(u)}\left(\sum_{w\in \neighbor{u}}\ket{\mathcal{R}_{u\gets w}}_{\mathcal R}\right) - \ket{\mathcal{R}_{u\gets v}}_{\mathcal R}\right]\right) \otimes \ket{\bot}^{\otimes |E|}_{\mathcal E}\\
        &\hspace{-14mm}= \left(\sum_{u\in V}\sum_{v \in \neighbor{u}}\sum_{w\in \neighbor{u}}\frac{2}{\deg(u)}\Phi_{u,v}(t-1)\ket{\mathcal{R}_{u\gets w}}_{\mathcal R} - \sum_{u\in V}\sum_{v \in \neighbor{u}}\Phi_{u,v}(t-1)\ket{\mathcal{R}_{u\gets v}}_{\mathcal R}\right) \otimes \ket{\bot}^{\otimes |E|}_{\mathcal E}\\
        &\hspace{-14mm}= \left(\sum_{u\in V}\sum_{v \in \neighbor{u}}\sum_{w\in \neighbor{u}}\frac{2}{\deg(u)}\Phi_{u,w}(t-1)\ket{\mathcal{R}_{u\gets v}}_{\mathcal R} - \sum_{u\in V}\sum_{v \in \neighbor{u}}\Phi_{u,v}(t-1)\ket{\mathcal{R}_{u\gets v}}_{\mathcal R}\right) \otimes \ket{\bot}^{\otimes |E|}_{\mathcal E}\\
        &\hspace{-14mm}= \left(\sum_{u\in V}\sum_{v \in \neighbor{u}} \left[\frac{2}{\deg(u)}\left(\sum_{w\in \neighbor{u}}\Phi_{u,w}(t-1)\right) - \Phi_{u,v}(t-1)\right]\ket{\mathcal{R}_{u\gets v}}_{\mathcal R}\right) \otimes \ket{\bot}^{\otimes |E|}_{\mathcal E}\\
        &\hspace{-14mm}= \left(\sum_{u\in V}\sum_{v \in \neighbor{u}} \Phi_{v,u}(t)\ket{\mathcal{R}_{u\gets v}}_{\mathcal R}\right) \otimes \ket{\bot}^{\otimes |E|}_{\mathcal E},
    \end{align*}
    where \cref{eq:rec} is used to derive the last equality.
     After Step \ref{step:swap}, the state becomes 
	\[
	\ket{\Psi''(t-1)} =\left(\sum_{u\in V}\sum_{v \in \neighbor{u}} \Phi_{v,u}(t)\ket{\mathcal{E}_{u\to v}}_{\mathcal E}\right) \otimes \ket{\bot}^{\otimes |E|}_{\mathcal R}.
	\]
    After Step \ref{step:communication}, the states becomes
    \begin{align*}
	\ket{\Psi'''(t-1)} &=\left(\sum_{u\in V}\sum_{v \in \neighbor{u}} \Phi_{v,u}(t)\ket{\mathcal{R}_{v\gets u}}_{\mathcal R}\right) \otimes \ket{\bot}^{\otimes |E|}_{\mathcal E}\\
    &=\left(\sum_{u\in V}\sum_{v \in \neighbor{u}} \Phi_{u,v}(t)\ket{\mathcal{R}_{u\gets v}}_{\mathcal R}\right) \otimes \ket{\bot}^{\otimes |E|}_{\mathcal E}\\
    &=\ket{\Psi(t)}.
    \end{align*}
    \paragraph{Message complexity.} The state just before communication happens at Step \ref{step:communication} of the loop is
    \[
        \left(\sum_{u\in V}\sum_{v \in \neighbor{u}} \Phi_{v,u}(t)\ket{\mathcal{E}_{u\to v}}_{\mathcal E}\right) \otimes \ket{\bot}^{\otimes |E|}_{\mathcal R}.
    \]
    This is a superposition of configurations each having only one register $\mathsf{E}_{u\to v}$ for some $(u,v)$ that is not equal to $\ket{\bot}$. Thus, the message complexity for one round is $b+1$ qubits. The message complexity for $T$ rounds is $(b+1)\cdot T$ qubits.
    \paragraph{Probability of success.}
    At Step \ref{step:measure}, the state is 
    \[
        \left(\sum_{u\in V}\sum_{v \in \neighbor{u}} \Phi_{u,v}(T)\ket{\mathcal{R}_{u\gets v}}_{\mathcal R}\right) \otimes \ket{\bot}^{\otimes |E|}_{\mathcal E}.
    \]
    We can use  \cref{lemma:QWcross} to analyze the success probability. From \cref{lemma:QWcross}, there exists $\hat T\in [2n, 3.6 n \log n]$ such that after $\hat T$ steps,  $$\sum_{v\in \neighbor{\exit}}|\Phi_{\exit, v}(\hat T)|^2 \geq \frac{1}{20 n}.$$
	The probability that $\exit$ outputs ``success'' is thus at least $\frac{1}{20n}$. On success, the register with outcome 1 contains $\ket{\tilde\varphi}$, and thus $\ket{\varphi}$.
\end{proof}

While $\Aa(T)$ allows to send a message from $\entr$ to $\exit$ with polynomial message complexity, two issues remain to be solved. First, the probability of success is only inverse-polynomial. Second, \cref{lemma:distrib_cross} only guarantees the existence of $ \hat T\in [2n, 3.6 n \log n]$ that achieves inverse-polynomial success probability but does not shows how to find this value. These issues can all be solved by calling $\Aa$ a polynomial number of times. This is implemented in the following algorithm, called $\Bb$. Since each call to $\Aa$ requires a new copy of $\ket{\qmess}$, $\Bb$ requires a polynomial number of copies of $\ket{\qmess}$. Of course, if the data is classical then only one copy is needed since classical information can be duplicated.

\begin{center}
		\begin{minipage}{16.4 cm} \vspace{2mm}
		\begin{algorithm}[H]
			\SetAlgorithmName{}{}{List of Algorithms}
			\nonl \hspace{-4mm}Algorithm $\Bb(n,\varepsilon)$ 
			\tcp*[l]{Quantum algorithm for point-to-point routing in $\Gg_n$}
            \nonl \textbf{Input:} $\entr$ receives several copies of the quantum data $\ket{\varphi} \in \mathbb C^{2^b}$ 
            
            \nonl\hspace{10mm}\tcp*[l]{only one copy is needed if the data is classical}\vspace{2mm}

            \For{$T$ {\bf from }$2n$ \KwTo $\lceil 3.6 n \log n \rceil$}
					{
                        \For{$k$ {\bf from }$1$ \KwTo $\lceil 20 n \log (1/\varepsilon) \rceil$}
                        {
                            Run $\Aa(T)$ on one new copy of the input $\ket{\varphi}$

                            \lIf{\rm ``success''}
				                {$\exit$ stops}
					    }
                    }
		\end{algorithm}
		\end{minipage}
	\end{center}

    When $T=\hat T$, the innermost loop of  $\Bb(n,\varepsilon)$ increases the success probability of Algorithm $\Aa(\hat T)$ to a constant by calling it $O(n\log 1/\varepsilon)$ times. The outermost loop tries all the possible values for $T$ from $2n$ to $3.6n\log n$. Note that $\Bb$ requires for all nodes to know the width $n$. However, it is possible to free ourselves from this constraint by using an additional loop trying every possible values for~$n$. This process keeps the complexity polynomial in $n$ but still requires a bound over $n$ to detect failures. 

    We now formally derive the success probability and the message complexity.
\begin{proposition}\label{prop:up}
	Algorithm $\Bb(n,\varepsilon)$ solves the point-to-point routing problem with probability at least $1-\varepsilon$ and message complexity $O(b\cdot\poly(n)\cdot\log(1/\varepsilon))$. 
\end{proposition}
\begin{proof}
	Let us assume that $\Bb(n,\varepsilon)$ is running on an instance $G$ of the welded trees $\Gg_n$ of width $n$. We note that $\Bb(n,\varepsilon)$ stops if and only if the message has reached $\exit$. Let us count all the calls to $\Aa$ :
	$$
	\#\texttt{Calls}(n) =  
    \sum_{T=2n}^{\lceil 3.6 n \log n \rceil}\sum_{k=1}^{\lceil 20 n \log (1/\varepsilon) \rceil} 1 =
	O\left(n^2\log n \log (1/\varepsilon)\right).
	$$
	This amounts to a polynomial number of calls to $\Aa(T)$, which has message complexity $O(b\cdot T)$ from Lemma \ref{lemma:distrib_cross}. Since $T$ is bounded by $\lceil 3.6 n \log n \rceil$, the final message complexity is 
    \[
    O\left(b\cdot n^3(\log n)^2 \log (1/\varepsilon)\right) = O(b\cdot\poly( n)\log(1/\varepsilon)).
    \]
	We bound the success probability as follows:
	\begin{align*}
		\mathbb{P}\left[\text{$\Bb(n,\varepsilon)$ succeeds}\right] &= \mathbb{P}\left[\text{$\Bb(n,\varepsilon)$ stops when $T=\hat T$}\right] \\
		&\qquad + \mathbb{P}\left[\text{$\Bb(n,\varepsilon)$ stops when $T\neq \hat T$}\right]\\
		&\geq \mathbb{P}\left[\text{$\Bb(n,\varepsilon)$ stops when $T=\hat T$}\right]\\
		&=1 - \mathbb{P}\left[\text{$\Bb(n,\varepsilon)$ does not stop when $T= \hat T$}\right]\\
		&=1 - \prod_{k=1}^{\lceil 20 n \log (1/\varepsilon) \rceil}\mathbb{P}\left[\text{$\Aa(\hat T)$ fails}\right]\\
		& \geq 1 - \prod_{k=1}^{\lceil 20 n \log (1/\varepsilon) \rceil}\left(1-\frac{1}{20 n}\right)\\
		&= 1- \left(1- \frac{1}{20n}\right)^{\lceil 20 n \log (1/\varepsilon) \rceil}\\
		&\ge 1- \varepsilon,
	\end{align*}
    as claimed.
\end{proof}

\cref{main_thm1} immediately follows from \cref{prop:up} by taking $\varepsilon=1/\exp(n)$.

\section{Classical Lower Bound}

This section is dedicated to the proof of Theorem~\ref{main_thm2}: the point-to-point routing problem over $\mathcal{G}_n$ has message complexity $2^{\Omega(n)}$, even if the nodes share randomness. In this section, $G=(V,E)$ denotes a graph in $\mathcal{G}_n$. For concreteness, we use the specific set of identifiers $\mathsf{ID} = [2^{n+2}-2]$ (note that a lower bound for this case immediately gives a lower bound for the general case where the identifiers can be arbitrary). 
Function $\id\colon V\rightarrow \mathsf{ID}$ matches each node to its identifier. Similarly, function $\rho\colon V\times[3]\rightarrow V$ matches each node and each of its port numbers taken in $[3]$ to the neighbor it communicates with. Let $\rho(\entr,3)=\rho(\exit,3)=\perp$.

The success probability and message complexity of an algorithm are defined for the graph $G$ in $\Gg_n$, the identifiers and the port numbers being chosen to be the worst possible ones for the algorithm to be run. However throughout the section, we assume that the latter are all drawn uniformly at random instead of being chosen as the worst case scenario. Since in this case the success probability is the mean success probability over all choices for $G$, $\id$ and $\rho$, this is necessarily an upper bound on the success probability in the worst case scenario.

Let us now introduce the following centralized game essentially equivalent to the point-to-point routing problem.

\begin{game}\label{game:simul}
    The player knows $\id(\entr)$, $\id(\exit)$ and has access to an oracle $O\colon\mathsf{ID}\times[3]\rightarrow\mathsf{ID}\times[3]$ such that $O(\id(v),p)$ returns $(\id(u),q)$ where $u=\rho(v,p)$ and $v=\rho(u,q)$.
    
    The  player must output a set of identifiers that corresponds to a path from $\entr$ to $\exit$.
\end{game}

The sources of randomness in this game, and the ones to come, are the graph $G$, the functions $\id$ and $\rho$, as well as the random coins used in the player's strategy. For any algorithm $\mathcal{A}$ for a game $X$, we define
\begin{equation*}
    \mathbb{P}_X(\mathcal{A}) = \mathbb{P}_{\substack{G,\id, \rho,\\\texttt{coins}}}\big[ \mathcal{A} \text{ wins Game X}  \text{ on } G \text{ with identifiers and port numbers given by } \id \text{ and }\rho\big].
\end{equation*}

Game 1 can simulate the behavior of any shared-randomness distributed algorithm in the network, and is thus easier than our original problem:
\begin{proposition}\label{prop:red1}
    If there exists an algorithm $\mathcal{A}$ that solves the point-to-point routing problem over $\Gg_n$ with probability at least $\alpha$ by sending $t$ messages, then there exists an algorithm $\mathcal{A}_1$ for Game~\ref{game:simul} using $t$ queries to the oracle such that
    \begin{equation*}
        \mathbb{P}_1(\mathcal{A}_1) \geq \alpha-\frac{1}{2^b}.
    \end{equation*}
    This lower bound holds even when algorithm $\mathcal{A}$ uses shared randomness. 
\end{proposition}
\begin{proof}
    Let $\mathcal{A}$ be such an algorithm for the point-to-point routing problem. We denote by $R$ its round complexity and for all $j\in [R]$, $M_j=\{(\id(v),p)\,|\,v\in V, p\in[3],m(j,\id(v),p)\neq\perp\}$ where $m(j,\id(v),p)$ is the message sent by node $v$ through its port number $p$ at round $j$, it is equal to $\perp$ if $v$ does not send a message through its port number $p$ at round $j$.

    First notice that the centralized player of Game~\ref{game:simul} can simulate the private coin of every distributed node, as well as their shared coin. As initially there is no node that has knowledge unavailable to the centralized player, the latter can first arbitrarily order $(\id(v),p)\in M_1$ and then compute $m(1,\id(v),p)$ for each of them. Then, $\mathcal{A}_1$ simulates the first round of message passing: for every $(\id(v),p)\in M_1$, $\mathcal{A}_1$ calls $O(\id(v),p)=(\id(u),q)$ and can suppose that node $u$ now has knowledge of $m(1,\id(v),p)$, as well as the fact that $\rho(u,q)=v$. After having done this for every $(\id(v),p)\in M_1$, $\mathcal{A}_1$ can compute $M_2$ and pursue accordingly to simulate $\mathcal{A}$ until the end. $\mathcal{A}_1$ has made exactly one call to $O$ for every $(\id(v),p)$ in $\bigcup_{1\leq j\leq R}M_j$, i.e., the query complexity of $\mathcal{A}_1$ is equal to $t$.

    Finally, two options follow the fact that $\mathcal{A}$ solves the point-to-point routing problem. The first option is that a node that received no information from $\entr$ randomly guessed the message~$\mess$ to route and (if it is not $\exit$ itself) transmitted it to $\exit$. This has probability at most $1/2^b$ of happening.
    The more sensible option is that $\mess$ has been forwarded from $\entr$ to $\exit$ during the run of $\mathcal{A}$. This means that the queries of the player in Game~\ref{game:simul} include a path between them, allowing the player to win the game.
\end{proof}

Let us now turn Game~\ref{game:simul} into an easier game.

\begin{game}\label{game:rdm_explo}
    As in Game~\ref{game:simul}, the player knows $\id(\entr)$, $\id(\exit)$ and has access to the oracle $O$. Let $t \geq 0$ be a parameter. Define $\Roots = \{ \id(\entr),\id(\exit) \}$, and for each $\id^r \in \Roots$, let $I(\id^r) = \{ \id^r \}$. For $i=1,\dots,t$, the player must choose $\id_i$ and $p_i$ according to the following instructions.\begin{itemize}
        \item[1a.] Either choose to draw $\id^r$ uniformly at random in $\mathsf{ID}\smallsetminus(\bigcup_{\id\in \Roots} I(\id))$, add $\id^r$ to $\Roots$ and set $I(\id^r)=\{\id^r\}$. Let $\id_i=\id^r$.
        \item[1b.] Or choose $\id^r\in \Roots$ and $\id_i\in I(\id^r)$.
        \item[2.] Let $P_i$ be the set of port numbers $p\in [3]$ such that\begin{itemize}
            \item $p\neq 3$ if $\id_i=\id(\entr)$ or $\id_i=\id(\exit)$,
            \item there is no $j<i$ such that $\id_j=\id_i$ and $p_j=p$,
            \item there is no $j<i$ such that $O(\id_j,p_j)=(\id_i,p)$.
        \end{itemize} 
        \item[3.] Draw $p_i$ uniformly at random in $P_i$ and add the identifier returned by $O(\id_i,p_i)$ to $I(\id^r)$.
    \end{itemize}
    The player wins if at least one of the following conditions holds:
    \begin{enumerate}
        \item there exists distinct $\id^r,\id^{r'}\in \Roots$ such that $I(\id^r)\cap I(\id^{r'})\neq \varnothing$, or,
        \item there exist $\id^r\in \Roots$ and $\id_i\in I(\id^r)$ such that $O(\id_i,p_i)$ returns $\id^r$ or the same identifier as $O(\id_j,p_j)$ for some $\id_j\in I(\id^r)\smallsetminus\{\id_i\}$.
    \end{enumerate}
\end{game}

Game~\ref{game:rdm_explo} explicitly simulates the random explorations that a distributed algorithm of message complexity $t$ performs. Every identifier in $\Roots$ corresponds to a node that independently decides to start exploring its surroundings by communicating with one of its neighbors. At the beginning, only $\entr$ and $\exit$ can hold this role. But at every step $i\in[t]$, any other node might independently decide to do the same. The player chooses Step 1a to simulate this. The randomness of drawing an identifier to add to $\Roots$ accounts for the fact that a node that has not yet received nor sent any messages has no information on which node of the graph it actually is when starting to explore. When the player rather chooses Step 1b, it signifies that the node of identifier $\id_i$, which was previously contacted during the exploration started by the node of identifier $\id_r$, extends the latter by sending a message through an unknown port.  The Step 2 is merely formalizing that the player is choosing an edge that has never been explored in one direction or an other which will be a dubious choice. It is in Step 3 that the message passing from the node of identifier $\id_i$ through its unknown port $p_i$ is simulated. The discovered node is then added to the set $I(\id^r)$ of nodes discovered during the exploration initiated by the one of identifier $\id^r$.

The first winning condition of the game represents two explorations crossing paths, while the second one represents an exploration looping back on itself. When one of these conditions is met and messages have been sent according to the simulation in the distributed algorithm, it can no longer be supposed that nodes \textit{obliviously} explore the network in the latter. Therefore, the instructions of Game~\ref{game:rdm_explo}, which account for this supposed obliviousness by randomizing choices of identifiers and port numbers, no longer accurately simulate the distributed algorithm. This is why we stop the simulation there and claim that the game is won. Nonetheless we will show that even those winning conditions are very unlikely to happen for any algorithm, through the final game, Game~\ref{game:embed}, essentially equivalent to Game~\ref{game:rdm_explo}.

For the moment, let us show that Game 2 is easier than Game 1.

\begin{proposition}\label{prop:red2}
    For any algorithm $\mathcal{A}_1$ for Game 1 that makes $t_1$ queries to the oracle, there exists an algorithm $\mathcal{A}_2$ for Game 2 with parameter $t_2\leq t_1$ such that
    \begin{equation*}
        \mathbb{P}_2(\mathcal{A}_2)\geq \mathbb{P}_1(\mathcal{A}_1).
    \end{equation*}
\end{proposition}
\begin{proof}
     Let us show that the lemma holds for $\mathcal{A}_2=\mathcal{A}_1$, as any query of $\mathcal{A}_1$ can be translated into Game~\ref{game:rdm_explo}. Denote $(O(\id_i,p_i))_{i\in[t_1]}$ the sequence of queries performed by Algorithm~$\mathcal{A}_1$. First, we ignore all queries of the following form:
     \begin{itemize}
        \item $O(\id(\entr),3))$ or $O(\id(\exit),3))$,
        \item $O(\id_i,p_i)$ if there is $j<i$ such that $\id_i=\id_j$ and $p_i=p_j$,
        \item $O(\id_i,p_i)$ if there is $j<i$ such that $O(\id_j,p_j)=(\id_i,p_i)$.
    \end{itemize}Indeed, in these queries the player already knows that no new information can be gained and thus removing them from the strategy of $\mathcal{A}_1$ does not affect its outcome. Denote $(O(\id_i,p_i))_{i\in[t_2]}$ the resulting sequence of queries. Let $i\in [t_2]$. Suppose that $(O(\id_j,p_j))_{1\leq j<i}$ also corresponds to the sequence of queries performed by $\mathcal{A}_2$ in Game~\ref{game:rdm_explo} so far. Suppose furthermore that for every $1\leq j <i$, $\id_j$ and the identifier returned by $O(\id_j,p_j)$ are both in $I(\id^r)$ for some $\id^r\in \Roots$.
    \begin{itemize}
        \item If $\id_i=\id(\entr)$ in $\mathcal{A}_1$, let the player of Game~\ref{game:rdm_explo} apply Step 1b, choose $\id^r=\id(\entr)\in \Roots$ and $\id_i=\id(\entr)\in I(\id(\entr))$. The player does the same respectively if $\id_i=\id(\exit)$. In this case, $\id_i$ in $\mathcal{A}_2$ does match $\id_i$ in $\mathcal{A}_1$.
        
        \item If there is no $j<i$ such that $\id_i=\id_j$ or $\id_i$ was returned by $O(\id_j,p_j)$ and if $\id_i$ is neither $\id(\entr)$ nor $\id(\exit)$, then $\id_i$ was in $\mathsf{ID}\smallsetminus(\bigcup_{\id\in \Roots} I(\id))$, the set of identifiers that were not encountered yet. Observe that, as the identifiers are assigned randomly to nodes, the strategy of $\mathcal{A}_1$ in this case is equivalent to choosing uniformly at random. This choice of $\id_i$ is then an application of Step 1a in Game~\ref{game:rdm_explo}.
        
        \item Otherwise, let $\id^r=\id_i$ and $i_r=i$. As long as there is $j<i_r$ such that $\id^r=\id'$ with $\id'=\id_j$ or $\id'$ was returned by $O(\id_j,p_j)$, set $\id^r=\id'$ and $i_r=j$. Notice that the resulting $\id^r$ falls into one of the two previous cases for $i=i_r$, meaning that $\id^r\in \Roots$. In the end this query was the result of applying Step 1b in Game~\ref{game:rdm_explo} and choosing $\id^r$ and $\id_i$ as defined here. By construction, one can verify that every intermediate $\id'$ was indeed added to $I(\id^r)$ in Step 3 of the corresponding turn $j\in [t]$.
    \end{itemize}
    Finally, it needs to be shown that the ports numbers $p_i$ for $i\in[t]$ were chosen exactly as stated in Steps 2 and 3 of Game~\ref{game:rdm_explo}. Notice that, as the strategy of $\mathcal{A}_1$ has been stripped of any useless queries, $p_i\in P_i$ for every $i\in[t]$, where $P_i$ is the set defined in Step 2 of Game~\ref{game:rdm_explo}. Similarly to the identifiers, as the port numbers are assigned randomly in $\mathcal{A}_1$, $p_i$ is chosen accordingly to Step 3 of Game~\ref{game:rdm_explo}.
    
    It has been shown that the queries of Game~\ref{game:simul} actually follow the steps of Game~\ref{game:rdm_explo}. Let us conclude by observing that the first winning condition of the latter is a straightforward requirement of linking $\entr$ and $\exit$ through a path in the former, from which follows the statement of the lemma. 
\end{proof}

We now define the random embedding of $r$-rooted trees. This will allow us to define Game~\ref{game:embed}, \textit{almost} equivalent to Game~\ref{game:rdm_explo}, but more intuitive and easier to manipulate. We start by defining $r$-rooted binary trees.

\begin{definition}
Let $r\ge 1$ be an integer. An $r$-rooted binary tree is a rooted tree such that the root has at most $r$ children and every other node has at most $2$ children. 
\end{definition}

Let $T$ be an $r$-rooted tree. We assume that the vertices of~$T$ are labeled as $\{0,\dots,t\}$ in such a way that $0$ is the label of the root, and if $i$ is the parent of $j$, then $i<j$. For a vertex $v\in G$ such that $r\le \deg_G(v)$, an embedding of $T$ in $G$ rooted in $v$ is a function $\pi\colon\{0,\dots,t\}\to V$ satisfying the following two conditions:
\begin{enumerate}
    \item $\pi(0) = v$,
    \item for $u,v$ neighbors in $T$, $\pi(u)$ and $\pi(v)$ are neighbors in $G$.
\end{enumerate}
We say that $\pi$ is improper if there exist two distinct $u,v\in T$ such that $\pi(u)=\pi(v)$.

\begin{definition}[random embedding]
The output $\pi\colon\{0,\dots,t\}\to V$ of the following process is called a random embedding of~$T$ in $G$ rooted in $v$:
\begin{enumerate}
    \item Set $\pi(0) = v$.
    \item 
    For each child $k\in [t]$ of \,$0$ in $T$, choose $\pi(k)$ uniformly at randomly from $\mathcal{N}_G(v)$ without replacement.
    \item For $i = 1$ to $t$, if $i$ is not a leaf of $T$ and if $\pi(i)$ is not $\entr$ nor $\exit$\footnote{If $\pi(i)=\entr$ or $\exit$, the rest of the embedding is not defined. This is not a problem in our case as a condition for winning Game~\ref{game:embed} is already met.}:
        \begin{enumerate}
            \item Let $j$ be the parent of $i$ in $T$;
            \item For each child $k\in[t]$ of $i$ in $T$, choose $\pi(k)$ uniformly at random from $\mathcal{N}_G(\pi(i))\setminus \{\pi(j)\}$ without replacement.
        \end{enumerate}
\end{enumerate}
\end{definition}

We are now ready to define a new game. 

\begin{game}\label{game:embed}
    Let $t\geq 0$, the player chooses $t+2$ trees $T_0, T_1,\ldots, T_{t+1}$ such that:
    \begin{enumerate}
        \item each tree has $t$ edges;
        \item $T_0$ and $T_1$ are 2-rooted binary trees;
        \item $T_2,\ldots, T_{t+1}$ are 3-rooted binary trees.
    \end{enumerate}
Set $v_0 = \entr$ and $v_1 = \exit$. Then for $i \in\{2,\ldots, t+1\}$, draw $v_i$ uniformly at random over $V\setminus \big(\{\entr,\exit\}\big)$ with replacement. For $i\in\{0,\dots,t+1\}$, let $\pi_i$ be a random embedding of $T_i$ in $G$ rooted in $v_i$.
The player wins if and only if there exist distinct $i,j \in\{0,\ldots, t+1\}$ such that $\pi_i(T_i)\cap\pi_j(T_j) \neq \varnothing$ or if there exists $i\in\{0,\ldots, t+1\}$ such that $\pi_i(T_i)$ is improper.
\end{game}
We first show that Game~\ref{game:embed} is easier than Game~\ref{game:rdm_explo}.

\begin{proposition}\label{prop:red3}
    Let Game~\ref{game:rdm_explo} and Game~\ref{game:embed} be parametrized by the same $t\geq 0$. Then, for any algorithm $\mathcal{A}_2$ for Game~\ref{game:rdm_explo}, there exists an algorithm $\mathcal{A}_3$ for Game~\ref{game:embed} such that
    \begin{equation*}
        \mathbb{P}_3(\mathcal{A}_3)\geq\mathbb{P}_2(\mathcal{A}_2)\,.
    \end{equation*}
\end{proposition}
\begin{proof}
    Once $\mathcal{A}_2$ has been run, define $G_r=(V_r,E_r)$ for $\id^r\in \Roots$, where $V_r$ is the set of vertices whose identifiers are in $I(\id^r)$, and $E_r$ is the set of edges incident to some node of identifier $\id_i$ through its port number $p_i$ such that $\id_i\in I(\id^r)$. Notice that as long as Game~\ref{game:rdm_explo} has not been won, the graphs $G_r$ are all trees of at most $t$ edges and are all node-disjoint one from another. Now define $\mathcal{A}_3$ the algorithm that constructs every tree $G_r$ in the same way as $\mathcal{A}_2$, assuming the game is not won until the end of the $t$-th iteration and without querying the different $O(\id_i,p_i)$. The only obstacle to doing so is that the choices of $\mathcal{A}_2$ might be dependent on the identifiers and port numbers that the queries return. However, as the latter are attributed randomly, $\mathcal{A}_3$ can instead use values it itself drew randomly in the correct supports of the functions $\id$ and $\rho$. 
    Finally, let $\mathcal{A}_3$ use the trees computed by $\mathcal{A}_2$ as its strategy for Game~\ref{game:embed}. It designates the nodes of identifier $\id^r \in \Roots$ as the roots of the trees ($\entr$ being the root of $T_0$ and $\exit$ that of $T_1$) and re-names the nodes as required by the embedding function $\pi$. For any tree of less than $t$ edges, $\mathcal{A}_3$ adds the missing edges arbitrarily. $\Roots$ being composed of $\id(\entr)$, $\id(\exit)$ and identifiers drawn uniformly at random without replacement at Step 1a for some turn $i\in [t]$ in Game~\ref{game:rdm_explo}, the trees computed by $\mathcal{A}_3$ respect the requirements of Game~\ref{game:embed}. We show that the claim of the lemma holds.
    
    By how $\mathcal{A}_3$ was just defined and by the definition of Game~\ref{game:embed}, as long as neither Game~\ref{game:rdm_explo} nor Game~\ref{game:embed} is won, for any $i\in [t]$, there is a one-to-one correspondence between any $\id_i$ considered in $\mathcal{A}_2$ and a node $\eta_i$ in some $\Pi_j$ resulting from $\mathcal{A}_3$ being run. Indeed, $\id_i$ corresponds to the node of same identifier in some $G_r$, which is renamed in some tree $T_j$ by $\mathcal{A}_3$ and then embedded into $\Pi_j$. 

    We now show that winning Game~\ref{game:embed} with algorithm $\mathcal{A}_3$ is at least as likely as winning Game~\ref{game:rdm_explo} with algorithm $\mathcal{A}_2$. Notice that for any query $O(\id_i,p_i)$ called from Steps 2 and 3 of Game~\ref{game:rdm_explo}, the latter are equivalent to adding a child to $\eta_i$ and embedding it through $\pi$. Indeed, the first restriction of Step~2 infers that if $T_2,\dots,T_{t+1}$ are 3-rooted, $T_0$ and $T_1$ are 2-rooted. The second restriction infers that two different children of the same node $\eta_i$ in any tree cannot be embedded to the same node in $V$. Finally, the third restriction infers that a child of $\eta_i$ cannot be embedded to the same node as the parent of $\eta_i$. The specification of the embedding $\pi$ in Game~\ref{game:embed} therefore follows exactly the same constraints as the choice of a port number in Game~\ref{game:rdm_explo}.

    This has established the equivalence between the choice of a port number in Game~\ref{game:rdm_explo} and the embedding of children in Game~\ref{game:embed}. Finally notice that this equivalence infers that the first winning condition of Game~\ref{game:rdm_explo} is at most as probable as there being distinct $i,j\in[t]$ such that $\Pi_i\cap\Pi_j\neq\varnothing$ in Game~\ref{game:embed} and that the second one is at most as probable as some $\Pi_i$ being improper for some $i\in[t]$. In fact, the winning probability of Game~\ref{game:embed} is bigger because of two reasons. First, a root of a tree $T_i$ corresponding to identifier $\id^r\in\Roots$ can be embedded to a node that corresponds to an identifier $\id_i$ of a query preceding $\id^r$ in Game~\ref{game:rdm_explo}. Whereas in this latter game, new identifiers in $\Roots$ are drawn in the set of not-yet-encountered identifiers. This slightly increases the winning probability of $\mathcal{A}_3$. Moreover, the extra edges added to the trees $T_j$ compared to the $G_r$ only increase the winning probability for Game~\ref{game:embed}.
\end{proof}



We now aim to establish that winning Game~\ref{game:embed} is very unlikely. To do so, we first introduce some notations. Let $\Pi_i = \pi_i(T_i)$ for $0\leq i\leq t+1$. For $k\in \{0,\dots,2n+1\}$, $C(k)$ denotes the $k^{th}$ column of $G$, i.e., the vertices of $G$ that are at distance $k$ of $\entr$. We extend the definition to $C(k,k')=\cup_{j=k}^{k'}C(j)$.
For any $v\in V$, $c(v)$ is the unique integer such that $v \in C(c(v))$.

For the remainder of the section $T_0,\dots,T_{t+1}$ are supposed to have already been fixed by the player. The only remaining sources of randomness are the functions $\id$ and $\rho$, which have no incidence on Game~\ref{game:embed}, the graph $G\in\Gg_n$ and the embeddings $\pi_0,\dots,\pi_{t+1}$. 
We then prove
several intermediate results on the behavior of the embeddings (Lemmas \ref{l:unif} -- \ref{l:5} below).

\begin{lemma}\label{l:unif}
        For any $i\in\{0,\dots,t+1\}$, $j\in T_i$ and $k \in \{0,\dots,2n+1\}$, $(\pi_i(j)\,|\,\pi_i(j)\in C(k))$ follows a uniform distribution over $C(k)$. 
\end{lemma}
\begin{proof}
    Let us prove this by induction on $j \in T_i$. First, node $\pi_i(0)$ is either $\entr$ or $\exit$, the two of which are alone in their respective columns, or $\pi_i(0)$ is drawn uniformly at random in $V \setminus \{\entr\,, \exit\}$. The lemma consequently holds for $j=0$.
        
    Let $j\in[t]$, assume that the lemma is true for any $j^\star\in\{0,\dots,j-1\}$. It is in particular true for the parent $j'$ and the grand parent $j''$ of $j$ in $T_i$. Then, for any $v\in C(k)$,
    \begin{align*}
        \mathbb{P}[\pi_i(j) = v] = \sum_{x \in \neighbor{v}} \mathbb{P}[\pi_i(j) = v | \pi_i(j') = x]\cdot\mathbb{P}[\pi_i(j') = x].
    \end{align*}
    Let us show that for any $x\in \neighbor{v}$ the corresponding term of the above sum does not depend on which node $v\in C(k)$ is considered. By assumption:
    \begin{align*}
        \mathbb{P}[\pi_i(j') = x] 
        &= \mathbb{P}[\pi_i(j')=x\,,\,\pi_i(j') \in C(c(x))]\\
        &= \mathbb{P}[\pi_i(j')=x\,|\,\pi_i(j') \in C(c(x))]\cdot \mathbb{P}[\pi_i(j') \in C(c(x))]\\
        &=\frac{1}{\lvert C(c(x))\rvert}\cdot\mathbb{P}[\pi_i(j') \in C(c(x))]\\
    \end{align*}
    Now observe that $\mathbb{P}[\pi_i(j) = v | \pi_i(j') = x]$ can either be equal to $\frac{1}{2}$ if $\pi_i(j'') \neq v$ or to $0$ if $ \pi_i(j'') = v$. Hence, once again using our assumption,
    \begin{align*}
        \mathbb{P}[\pi_i(j) = v \,|\, \pi_i(j') = x] 
        &= \frac{1}{2}\cdot\mathbb{P}[\pi_i(j'') \neq v]\\
        &= \frac{1}{2}\cdot\big(1-\mathbb{P}[\pi_i(j'') = v]\big)\\
        &= \frac{1}{2}\cdot\left(1-\frac{\mathbb{P}[\pi_i(j'') \in C(k)]}{\lvert C(k) \rvert}\right).
    \end{align*}
    It has now been proven that for every $x\in\neighbor{v}$, $\mathbb{P}[\pi_i(j) = v | \pi_i(j') = x]\cdot\mathbb{P}[\pi_i(j') = x]$ only depends on $C(k)$ and $C(c(x))$. Finally, as every vertex in $C(k)$ has exactly the same number of neighbors in each column it follows that for any $v,v' \in C(k)$:
    \begin{equation*}
        \mathbb{P}[\pi_i(j) = v] = \mathbb{P}[\pi_i(j) = v']
    \end{equation*}
    Therefore, $(\pi_i(j) | \pi_i(j) \in C(k))$ follows a uniform distribution over $C(k)$.        
\end{proof}

\begin{lemma}\label{l:middle}
        Let $i,j\in\{0,\dots, t+1\},i\neq j$. Then $\mathbb{P}\left[\Pi_i\cap\Pi_j\cap C(\frac{n}{2}+1,\frac{3n}{2}) \neq \varnothing\right] = O\left(\frac{t^2}{2^{n/2}}\right)$.
    \end{lemma}
    \begin{proof}
        Let $x \in T_i$ and $y \in T_j$. Assuming that there exists $\frac{n}{2}+1 \leq k \leq n+\frac{n}{2}$ such that $c(x)=c(y) = k$. Then by Lemma ~\ref{l:unif} and the independence of the random variables $\pi_i(x)$ and $\pi_j(y)$,
        \[
            \mathbb{P}[\pi_i(x) = \pi_j(y) | c(\pi_i(x)) = c(\pi_j(y)) = k] = \frac{1}{\lvert C(k)\rvert} \leq \frac{1}{2^{n/2}}.
        \]
        Conclude by applying the union bound on the $O(t^2)$ different pairs $(x,y)\in\{0,\dots,t\}^2$.
    \end{proof}
    
    \begin{lemma}\label{l:extremity}
        Let $i\in\{0,\dots,t+1\}$. If $c(v_i) \geq \frac{3n}{4}$ then $\mathbb{P}\left[\Pi_i \cap C(0,\frac{n}{2}) \neq \varnothing\right]\le\frac{t^2}{2^{n/4}}$. Similarly, if $c(v_i)\leq \frac{5n}{4}+1$, then $\mathbb{P}\left[\Pi_i\cap C(\frac{3n}{2}+1,2n+1) \neq \varnothing\right]\le\frac{t^2}{2^{n/4}}$.
    \end{lemma}
    \begin{proof}
        Let us prove the first statement, the second one holds by symmetry.
        Consider a path $P$ in $T_i$ from $v_i$ to a leaf. As $G\in\Gg_n$, it is necessary for $\pi_i(P)$ to move left at least $\frac{n}{4}$ times in a row in order to reach $C(\frac{n}{2})$ from $C(c(v_i))$ and this trial succeeds with probability $\frac{1}{2^{n/4}}$. Now, each path $P$ in $T_i$ from $v_i$ to a leaf has at most $t$ tries and there are at most $t$ such paths $P$ in $T_i$. The lemma follows from the union bound.
    \end{proof}
    
    Let us now bound the probability that two embedded trees intersect each other.
    
    \begin{lemma}\label{l:4}
         $\mathbb{P}\left[\textrm{there exist distinct } i,j\in\{0,\dots,t+1\}\textrm{ such that }\, \Pi_i \cap \Pi_j \neq \varnothing\right] = O\left(t^4\cdot2^{-{n/4}}\right).$
    \end{lemma}
\begin{proof}
    Let us consider all possible points of intersection for $\Pi_i$ and $\Pi_j$. First, by Lemma~\ref{l:middle}, $\Pi_i$ and $\Pi_j$ have probability $O(\frac{t^2}{2^{n/2}})$ of intersecting each other in $C(\frac{n}{2}+1,\frac{3n}{2})$.

    Suppose now that they intersect at $C(0,\frac{n}{2})$ instead. Without loss of generality, $\Pi_i$ has the leftmost root of the two trees. As it is in a column to the right of $v_i$ (or the same), $v_j$ cannot be $\entr$. This means that $v_j$ has either been drawn uniformly at random in $C(1,2n)$, or $v_j$ is $\exit$. Thus, 
    $$\mathbb{P}\left[0\leq c(v_j)\leq \frac{3n}{4}-1\right]\leq\frac{\sum_{k=1}^{(3n/4)-1}2^k}{2\sum_{k=1}^n2^k}=\frac{2^{{3n/4}}-2}{2(2^{n+1}-2)}<2^{-{n/4}}.$$ 
    Moreover, by Lemma~\ref{l:extremity}, if $c(v_j)\geq \frac{3n}{4}$, then the probability that $\Pi_j$ intersects $C(0,\frac{n}{2})$ is at most $\frac{t^2}{2^{n/4}}$. Therefore, $\Pi_i$ and $\Pi_j$ have probability $O(t^2\cdot 2^{-n/4})$ of intersecting each other in $C(0,\frac{n}{2})$. By symmetry, they have no greater probability of intersecting in $C(\frac{3n}{2}+1,2n+1)$.

    In the end, any two trees $\Pi_i$ and $\Pi_j$ have probability $O(t^2\cdot(2^{-n/2}+2^{-n/4}))=O(t^2\cdot 2^{-n/4})$ of intersecting each other. The result is once again yielded by a union bound, this time over the $O(t^2)$ different pairs $(i,j)\in\{0,\dots,t+1\}^2$.
\end{proof}

\begin{lemma}\label{l:5}
    $\mathbb{P}\left[\exists i\in\{0,\dots,t+1\},\pi_i \textrm{ is an improper embedding of }T_i\right]=O\left(t^5\cdot 2^{-n/4}\right)$.
\end{lemma}
\begin{proof}
    Consider distinct $x,y \in T_i$, let us bound the probability that $\pi_i(x) = \pi_i(y)$. Suppose that the event occurred and denote $P$ the path linking $x$ and $y$ in $T_i$. Then its embedding $\pi_i(P)$ contains a cycle in $G$. Observe that a cycle in $G$ must intersect both $C(n)$ and $C(n+1)$. Let $z$ be a vertex of $P$ such that $\pi_i(z)\in C(n)$, $P_1$ the path from $z$ to $x$ in $T_i$ and $P_2$ the path from $z$ to $y$ in $T_i$. Assume without loss of generality that $\pi_i(P_1)$ intersects $C(n+1)$. If $c(\pi_i(x))\leq \frac{n}{2}$, then Lemma ~\ref{l:extremity} states that the probability  that $\pi_i(P_1)$ reaches $C(\frac{n}{2})$ after visiting column $n+1$ is at most $\frac{t^2}{2^{n/4}}$. The same goes respectively if $c(\pi_i(x)) \geq \frac{3n}{2}+1$.
    
    Let us then focus on the probability that $\pi_i(x)=\pi_i(y)$ knowing that $\frac{n}{2}+1\leq c(\pi_i(x))\leq \frac{3n}{2}$. The graph induced by $C(\frac{n}{2}+1,n)$ is made of $2^{(n/2)+1}$ disconnected trees. Let us call them $S_1,\dots,S_{2^{(n/2)+1}}$. And $S'_1,\dots,S'_{2^{(n/2)+1}}$ denote the disconnected trees of the graph induced by $C(n+1,\frac{3n}{2})$. Path $\pi_i(P_1)$ visits a sequence $S_{i_1},S'_{j_1},S_{i_2},\dots$ of the subtrees defined before. Similarly $\pi_i(P_2)$ visits $S_{k_1},S'_{l_1},S_{k_2},\dots$. In the case where $\pi_i(x)=\pi_i(y)$, the last term of the two sequences must coincide. Let us bound the probability of that happening. Note that the sequence of trees for $\pi_i(P_1)$ has at least two terms. Fix $\pi_i(P_1)$ up until its jump from the second-to-last to the last term of its list of visited subtrees. The embedding of the next edge belongs to the random cycle of $G$, it has uniform probability of leading to any leaf of the other side that is not incident to an edge both in $\pi_i(P_1)$ and the random cycle. There are at least $2^n-t$ such leaves, at most $2^{n/2}$ of which belong to the last visited subtree of $\pi_i(P_2)$. Therefore, the probability that the last term on the two lists of visited subtrees by $\pi_i(P_1)$ and $\pi_i(P_2)$ coincide is bounded by $\frac{2^{n/2}}{2^n-t}=O(\frac{t}{2^{n/2}})=O(\frac{t^2}{2^{n/4}})$.
        Yet another union bound, on the $\binom{t+1}{2}$ possible distinct pairs $x,y\in T_i$ for $i\in\{0,\dots,t+1\}$, allows us to conclude.
    \end{proof}

We are now ready to establish that winning Game~\ref{game:embed} is very unlikely:
\begin{proposition}\label{propo:lb}
     There is no strategy that allows the player to win Game ~\ref{game:embed} with parameter $t$ with probability more than $O\left(t^5\cdot 2^{-n/4}\right)$.
\end{proposition}
\begin{proof}
    Lemmas \ref{l:4} and \ref{l:5} directly infer that for any algorithm $\mathcal{A}_3$ choosing trees $T_0,\dots,T_{t+1}$,
    \begin{equation*}\label{prop:lb}
        \mathbb{P}_3(\mathcal{A}_3) = O\left(\frac{t^5}{2^{n/4}}\right).
    \end{equation*}
\end{proof}

From Propositions \ref{prop:red1}, \ref{prop:red2}, \ref{prop:red3} and \ref{propo:lb}, it follows that any distributed algorithm solving the point-to-point routing problem with $t$ messages has success probability $O\left(t^5\cdot 2^{-n/4}+2^{-b}\right)$. \cref{main_thm2} is the contraposition of this latter statement.

\bibliographystyle{plain}
\bibliography{references}
\end{document}